\documentclass[submission,copyright,creativecommons]{eptcs}
\providecommand{\event}{DCM 2015} % Name of the event you are submitting to
\usepackage{breakurl}             % Not needed if you use pdflatex only.
\usepackage{amssymb}
\usepackage{graphicx}
\usepackage{epsfig}
\usepackage{amsmath}
\usepackage{amssymb}
\usepackage{algorithm}
\usepackage{algorithmic}
\usepackage{url}

\newcommand{\ie}{\textit{i}.\textit{e}.,\ }
\newcommand{\eg}{\textit{e}.\textit{g}.\ }
\newtheorem{predefin}{{\bf Definition}}

\newenvironment{definition}{\begin{predefin}{\hspace{-0.5 em}}
               }{\end{predefin}}
 \newtheorem{prelem}{{\bf Lemma}}

\newtheorem{prealphthm}{{\bf Theorem}}

 \newtheorem{prethm}{{\bf Theorem}}
 
 \newenvironment{theorem}{\begin{prethm}{\hspace{-0.5
               em}}}{\end{prethm}}

\newtheorem{preproof}{{\bf Proof.}}

\newenvironment{proof}[1]{\begin{preproof}{\rm
               #1}\hfill{$\blacksquare$}}{\end{preproof}}
\title{Generation, Ranking and Unranking of Ordered Trees with Degree Bounds}
\author{Mahdi Amani
\institute{Dipartimento di Informatica\\ Universit$\grave{\text{a}}$ di Pisa, Pisa, Italy.}
\email{m\_amani@di.unipi.it}
\and
 Abbas Nowzari-Dalini
\institute{School of Mathematics, Statistics, and Computer Science\\
 Colleague of Science,  University of Tehran, Tehran, Iran.}
\email{\quad nowzari@ut.ac.ir}
}
\def\titlerunning{Generation, Ranking and Unranking of Ordered Trees with Degree Bounds}
\def\authorrunning{M. Amani, A. Nowzari-Dalini}
\begin{document}
\maketitle

\begin{abstract}
We study the problem of generating, ranking and unranking of unlabeled ordered trees whose nodes have maximum degree of $\Delta$. This class of trees represents a generalization of chemical trees. A chemical tree is an unlabeled tree in which no node has degree greater than 4. By allowing up to $\Delta$ children for each node of chemical tree instead of 4, we will have a generalization of chemical trees. Here,  we introduce a new encoding over an alphabet of size 4 for representing unlabeled ordered trees with maximum degree of $\Delta$.  We use this encoding for generating these trees in A-order with constant average time and $O(n)$ worst case time. Due to the given encoding, with a precomputation of size and time $O(n^2)$ (assuming $\Delta$ is constant), both ranking and unranking algorithms are also designed taking $O(n)$ and $O(n \log n)$ time complexities.
\end{abstract}

% \begin{keywords}
%Ordered tree with degree bounds, Chemical tree, Generation algorithm, Ranking and unranking, Restricted Tree.
% \end{keywords}

%--------------------------------Intro-----------------------------------------------
\section{Introduction}
\noindent 
A {\em labeled tree} is a tree in which each node is given a unique label. 
A {\em rooted tree} is a tree in which one of the nodes is distinguished from the others as the {\em root}. An {\em ordered tree} or {\em plane tree} is a rooted tree for which an ordering is specified for the children of each node.  %Given an embedding of a rooted tree in the plane, if one fixes a direction of children, say left to right, then an embedding gives an ordering of the children. Conversely, given an ordered tree, and conventionally drawing the root at the top, then the child nodes in an ordered tree can be drawn left-to-right, yielding an essentially unique planar embedding.
Studying combinatorial properties of restricted graphs or graphs with configurations has many applications in various fields such as machine learning and chemoinformatics. Studying combinatorial properties of restricted trees and outerplanar graphs (\eg ordered trees with bounded degrees) can be used for many purposes including virtual exploration of chemical universe, % [15],
reconstruction of molecular structures from their signatures, % [9], 
and the inference of structures of chemical compounds~\cite{Ari03,Fuj08,Gut86,Han94,Shi11,Gut87,Zhu10}. % [7, 12]. 
In this paper we study the generation of unlabeled ordered trees whose nodes have  maximum degree of $\Delta$ and for the sake of simplicity, we denote it by {\em $T^\Delta$ tree}, also we use $T_n^\Delta$ to denote the class of $T^\Delta$ trees with $n$ nodes.  

Chemical trees are the most similar trees to $T^\Delta$ trees. 
Chemical trees are the graph representations of {\em alkanes}, or more precisely, the carbon atom skeleton of the molecules of alkanes~\cite{CapGut99,DobGut99,Gut95,Gut86,Gut98,Gut87}.
The alkane molecular family is partitioned into {\em classes of homologous molecules},
that is molecules with the same numbers of carbonium and hydrogen atoms; the
$n^{th}$ class of alkane molecular family is characterized by the formula $C_nH_{2n+2}$, $n = 1, 2, ...$~\cite{Ari03} with the same numbers of carbonium
and hydrogen atoms.
They are usually represented by indicating the carbonium atoms
and their  links, omitting to represent hydrogen atoms~\cite{Ari03}; therefore, all the nodes would have the same label; carbon (\ie  the tree is unlabeled), as shown in Figure~\ref{alkanepic} for $n=3$ and $n=4$.
  A {\em chemical tree} is defined as a tree in which no node has degree greater than 4~\cite{Gut87,Gut86,Gut95,DobGut99,CapGut99,Gut98}, chemical trees are also considered to be unlabeled~\cite{Gut95,DobGut99,CapGut99,Gut98}. Therefore, $T^\Delta$ tree can be considered as a generalization of chemical trees to   unlabeled ordered trees whose nodes have maximum degree of $\Delta$ instead of 4.

\begin{figure}
\begin{center}
\includegraphics[scale=0.25]{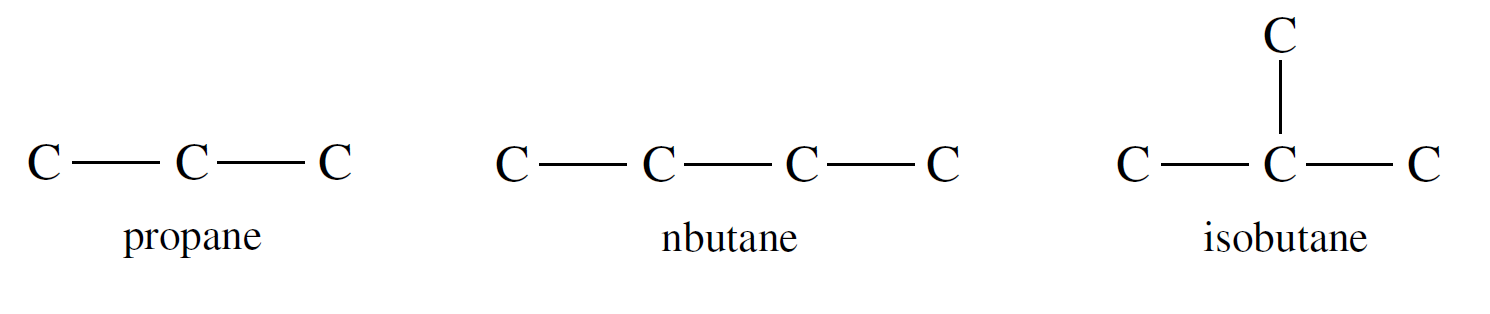}
\end{center}
\caption{Left: $C_3H_8$ propane, middle and right: $C_4H_{10}$ butanes.}
\label{alkanepic}
\end{figure}

 Generation, ranking and unranking of trees are very basic problems in computer science and discrete mathematics~\cite{za80}.   
 In general, for any combinatorial object of size $n$, we can define a variety of orderings.  Since we study trees, let's consider an arbitrary class of trees of size $n$ ($n$ nodes) showed by $\mathbb{T}_n$ (\textit{e.g.} semi-chemical trees of size $n$), the elements  of this set ($\mathbb{T}_n$) can be listed based on any defined ordering. Two such orderings for the trees are {\em A-order} and {\em B-order}, which will be explained in the next section.
By having $\mathbb{T}_n$ and an ordering (\eg A-order or B-order), {\em next}, {\em rank}, and {\em unrank} functions are defined as follows.

For a given tree $T\in \mathbb{T}_n$, the {\em next function} gives the successor tree of $T$ with respect to the defined ordering, the `position' of tree $T$ in $\mathbb{T}_n$ is called {\em rank}, the {\em rank function} determines the rank of $T$; the inverse operation of ranking is {\em unranking}, for a position $r$, the {\em unrank function} gives the tree $T$ corresponding to this position. A {\em generation algorithm}, generates all the elements in $\mathbb{T}_n$ with respect to the given ordering (starting from the first tree, then repeating the next function until producing the last tree)~\cite{za80}.

In most of the tree generation algorithms, a tree is represented by an integer or an alphabet sequence called {\em codeword}, hence all possible sequences of this representation are generated. 
This operation is called {\em tree encoding}. 
 Basically, the uniqueness of the encoding, the length of the encoding, and the capability of constructing the tree from its representation, which is called {\em decoding}, are essential considerations in the design of the tree encoding schema~\cite{za80}. 
 
  Many papers have been published earlier in the literature  for generating different classes of trees.  For example we can mention the generation of binary trees 
 in~\cite{para85,luroru93,vapa94,AhNo98,ahno99,ahno05,wuchawan06},
 $k$-ary trees    in~\cite{ru78,er92,vapa97,korlaf99,xiaushtan01,kor05,HeuLiMan08,AdlNowAhr11,WuChaCha11}, 
rooted trees in~\cite{Bey80,Wil89}, %Nak03,
  trees with $n$ nodes and $m$ leaves in~\cite{pa87,TabAhrNow10},
neuronal trees in~\cite{Pal90},  %amani15neu
  and AVL trees in~\cite{li86}. 
  On the other hand,  many papers have thoroughly investigated basic combinatorial features of chemical  trees~\cite{Gut87,Gut86,Gut95,DobGut99,CapGut99,Gut98,Will98}.  

    More related to our work, in~\cite{Han94} a coding for chemical trees without the generation algorithm, and in~\cite{Ari03} the enumeration of chemical trees and in~\cite{Fuj08,Shi11} the enumeration of tree-like chemical graphs have been presented.
Hendrickson and Parks in~\cite{Hen91} investigated the enumeration and the generation of carbon skeletons which can have cycles and are not necessarily trees. 
The work most related to our paper is an algorithm for the generation of certain classes of trees such as chemical trees in~\cite{Bal92} with no ranking or unranking algorithm.   
% In enumeration of chemical ~\cite{,Ari03} 
% enumeration of tree-like chemical graphs ~\cite{Fuj08,Shi11}
% in~\cite{Hen91} they studied enumeration and generation of Carbon skeletons which can have cycles and are not necessarily trees.
% Generation  of chemical trees   in~\cite{Bal92} they presented a generation of certain
  In that paper, all chemical trees with $n$ nodes are generated from the complete set of chemical trees with $n-1$ nodes, the redundant generations are possible and they needed to minimize the possible redundancy. 
  
  The problem of enumeration of ordered trees were also studied in~\cite{Yam09} and the generation of different ordered trees (with no bounds on the degrees of the nodes) were studied in~\cite{za80}. In~\cite{Zhu10}, a generation algorithm with constant average delay time but with no ranking or unranking algorithms was given for all unrooted trees of $n$ nodes and a diameter at least $d$ such that the degree of each node with distance $k$ from the center of the tree is bounded by a given function. 
  In~\cite{Wri09} all unrooted unlabeled trees have been generated in constant average time with no ranking or unranking algorithms.  Nakano and Uno in~\cite{Nak04} gave an algorithm to generate all rooted unordered trees with exactly $n$ nodes and diameter $d$ in constant delay time.  Up to now, to our knowledge,  no efficient generation, ranking or unranking algorithms are known for either `chemical trees' or `ordered trees with bounded degrees'.

%The A-order definition uses global information
%concerning the nodes and appear to be a natural ordering of trees, whereas the B-order uses local information.

The remaining of the paper is organized as follows. Section~\ref{Defsection} introduces the definitions and notions that are used further. Our new encoding for 
$T_n^\Delta$ trees is presented in Section~\ref{Encodingsection}. % ($T_n^{\,4}$ represents the chemical tress with $n$ nodes).
 The size of our encoding is $n$ while the alphabet size is always 4. Based on the presented encoding, a new generation algorithm with constant average time  and $O(n)$ worst case time is given Section~\ref{nextchap}. In this algorithm, $T_n^\Delta$ trees are generated in A-order. Ranking and unranking algorithms are also designed in Section~\ref{Ranksection} with $O(n)$ and $O(n \log n)$
 time complexities, respectively. The presented ranking and unranking algorithms need a precomputation of size and time $O(n^2)$ (assuming $\Delta$ is constant).

%------------------------------------Defini---------------------------------------
\section{Definition}
\label{Defsection}
\noindent

%In this section, we give some more detailed definitions.
As mentioned before, the $n^{th}$ class of alkane molecular family is characterized by the formula $C_nH_{2n+2}$, $n = 1, 2, ...$~\cite{Ari03} with the same numbers of carbonium and hydrogen atoms. They are usually represented by indicating the carbonium atoms and their links, omitting to represent hydrogen atoms. Therefore, a chemical tree is a tree in which no node has degree greater than 4~\cite{Gut87,CapGut99,Gut98}. 
%Chemical trees with exactly $n$ nodes (internal nodes and leaves) are called {\em $n$-node chemical trees}. For $n\leq 5$, all of the $n$-node chemical trees (for example butane or 2-methylpropane) are also the $n$-node trees, but for larger values of $n$ the number of $n$-node chemical trees is significantly smaller than the number of $n$-node trees~\cite{Gut87,Gut86,DobGut99,CapGut99,Gut98,Will98}.  
%%Gut95%Gut86, DobGut99
% % Chemical trees with exactly $n$ nodes 
%%are also called {\em n-vertex chemical trees}. For $n\leq5$, all of the {\em n-vertex chemical trees} %(for example butane or 2- methylpropane) 
%%are also the $n$-vertex trees, but for larger values of $n$ the number of $n$-vertex chemical trees is
%%significantly smaller than the number of $n$-vertex trees~\cite{Gut95, Gut86, DobGut99, CapGut99, Gut87}.$T^\Delta$
We study the class of ordered trees whose nodes have maximum degree of $\Delta$ and for the sake of simplicity, we denote it by {\em $T^\Delta$}. $T^\Delta$
can be considered as a generalized version of chemical trees. % by letting them be ordered and allowing nodes to have degrees up to $\Delta$.

 Formally, a $T^\Delta$ tree $T$ is defined as a finite set of nodes such that $T$ has a root $r$, and if $T$ has more than one node, $r$ is connected to $j \leq \Delta$  subtrees $T_1, T_2, \ldots , T_j$ which each one of them is also recursively a $T^\Delta$ tree and by $T_n^\Delta$ we represent the class of $T^\Delta$ trees with $n$ nodes. An example of a  $T^\Delta$ tree is shown in Figure~\ref{chemexample}.

\begin{figure}
\begin{center}
\includegraphics[scale=0.9]{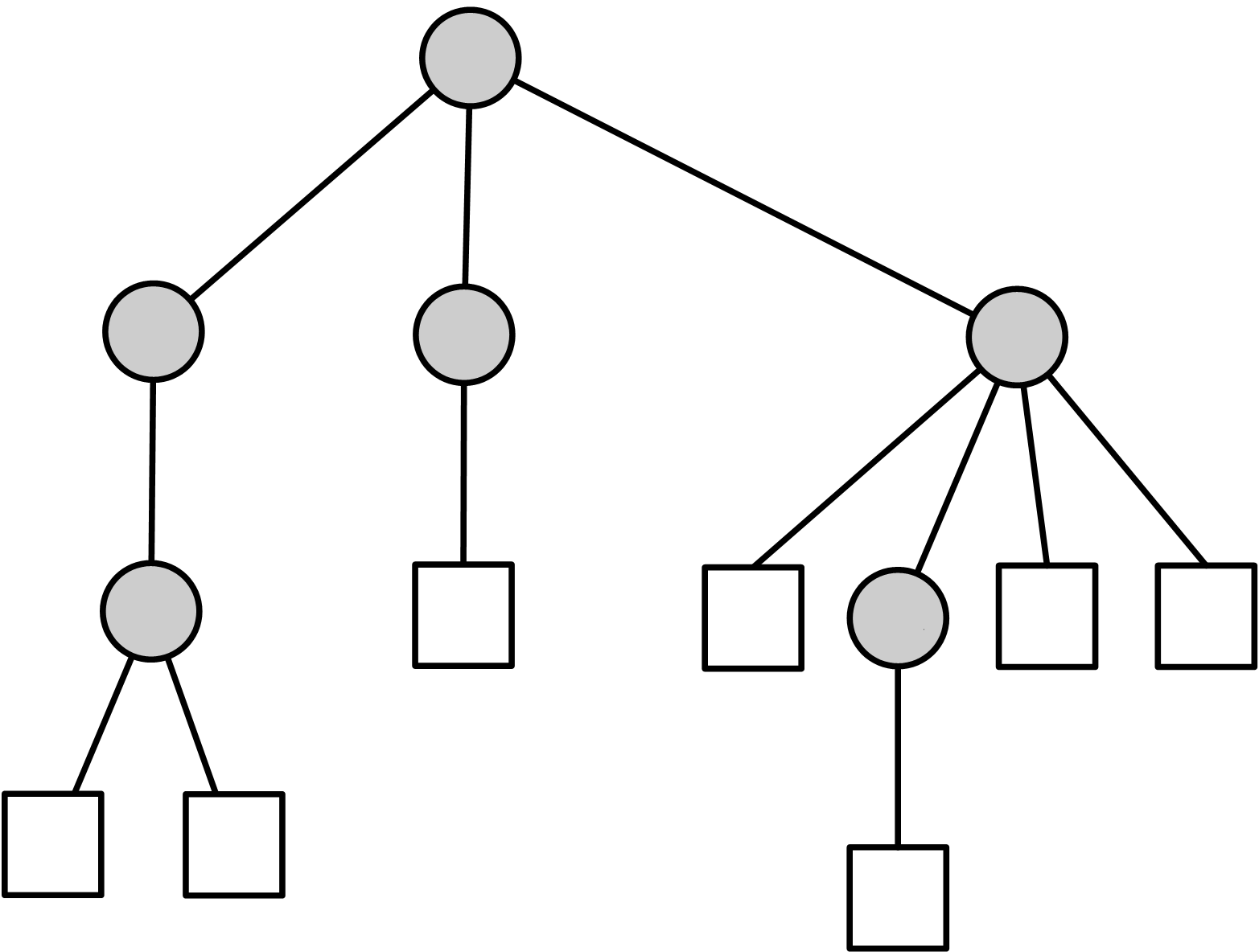}
\end{center}
\caption{A $T^\Delta$ tree with 13 nodes.}
\label{chemexample}
\end{figure}

As mentioned earlier, in  most of the tree  generation algorithms, a 
 tree is represented by an integer or an alphabet sequence called {\em codeword},  hence, all possible sequences of this representation are generated.
   In general, on any class of trees, we can define a variety of ordering
 for the set of trees.
Classical orderings on trees are
 {\em A-order} and {\em B-order} which are defined as follows~\cite{vapa94,za80}.

\begin{definition}
Let $T$ and $T'$ be two trees in $T^\Delta$ and $k= max
\{deg(T),deg(T')\}$, we say that $T$ is less than $T'$ in A-order
$(T \prec_{A} T')$, iff
\begin{itemize}

    \item $|T| < |T'|$, or

    \item $|T| = |T'|$ and for some $1 \leq i \leq k,$
     $T_j =_A T'_j$ for all $j = 1, 2, \ldots, i-1$ and $T_i \prec_A T'_i$;
\end{itemize}
 where $|T|$  is the number of nodes in the tree $T$ and $deg(T)$ is the degree of the root of $T$.
\end{definition}

\begin{definition}
Let $T$ and $T'$ be two trees in $T^\Delta$ and $k= max
\{deg(T),deg(T')\}$, we say that $T$ is less than $T'$ in B-order
$(T \prec_{B} T')$, iff
\begin{itemize}

    \item $deg(T) < deg(T')$, or

    \item $deg(T) = deg(T')$ and for some $1 \leq i \leq k,$
     $T_j =_B T'_j$ for all $j = 1, 2, \ldots, i-1$ and $T_i \prec_B T'_i$.
\end{itemize}
\end{definition}

 Our generation algorithm, which is given in the Section~\ref{nextchap},
 produces the sequences corresponding to $T_n^\Delta$ trees in
 A-order. 
 For a given tree $T\in T_n^\Delta$, the {\em generation algorithm} generates all the  successor trees of $T$ in $T_n^\Delta$; the position of tree $T$ in $T_n^\Delta$ is called {\em rank}, the {\em rank function} determines the rank of $T$; the inverse operation of ranking is {\em unranking}. These functions can be also employed in any random generation of $T_n^\Delta$ trees, for example.

\section{The encoding schema}
\label{Encodingsection}
\noindent 
The main point in generating trees is to choose a suitable encoding
 to represent them, and generate their corresponding codewords.
  Regarding the properties of $T_n^\Delta$, we present our new encoding. For any tree $T\in T_n^\Delta$, the encoding over 4 letters $\{s, \ell, m, r\}$ is defined as follows. The root of $T$ is labeled by $s$, and for any internal node, if it has only one child, that child is labeled by $s$, otherwise the leftmost child is labeled by $\ell$, and the rightmost child is labeled by $r$, and the children between the leftmost and the rightmost children (if exist) are all labeled by $m$. Nodes are labeled in the same way for any internal node in each level recursively, and by a pre-order traversal of $T$, the codeword will be obtained (one can say the labels $\ell$ and $r$ operate as left and right parenthesis to define the fingerprint of a subtree in the codewords, and $s$ is used when the left and right bounds are the same). This labeling is illustrated in Figure~\ref{chemcoded}.
Using this encoding, the 4-letters alphabet codeword corresponding to the first and last $T_n^\Delta$ trees in A-order are respectively ``$s\ell m^{\Delta-2}r\ell m^{\Delta-2}r \ldots \ell m^{(n\ mod\ \Delta)-2}r$'' and ``$s^{n}$'' which are shown in the Figure~\ref{chemfirstlast}-a and Figure~\ref{chemfirstlast}-b. Now, we prove the validity of this encoding for $T_n^\Delta$ trees (one-to-one correspondence). %A valid codeword is defined by using the following definition and theorem.

\begin{figure}
\begin{center}
\includegraphics[scale=0.9]{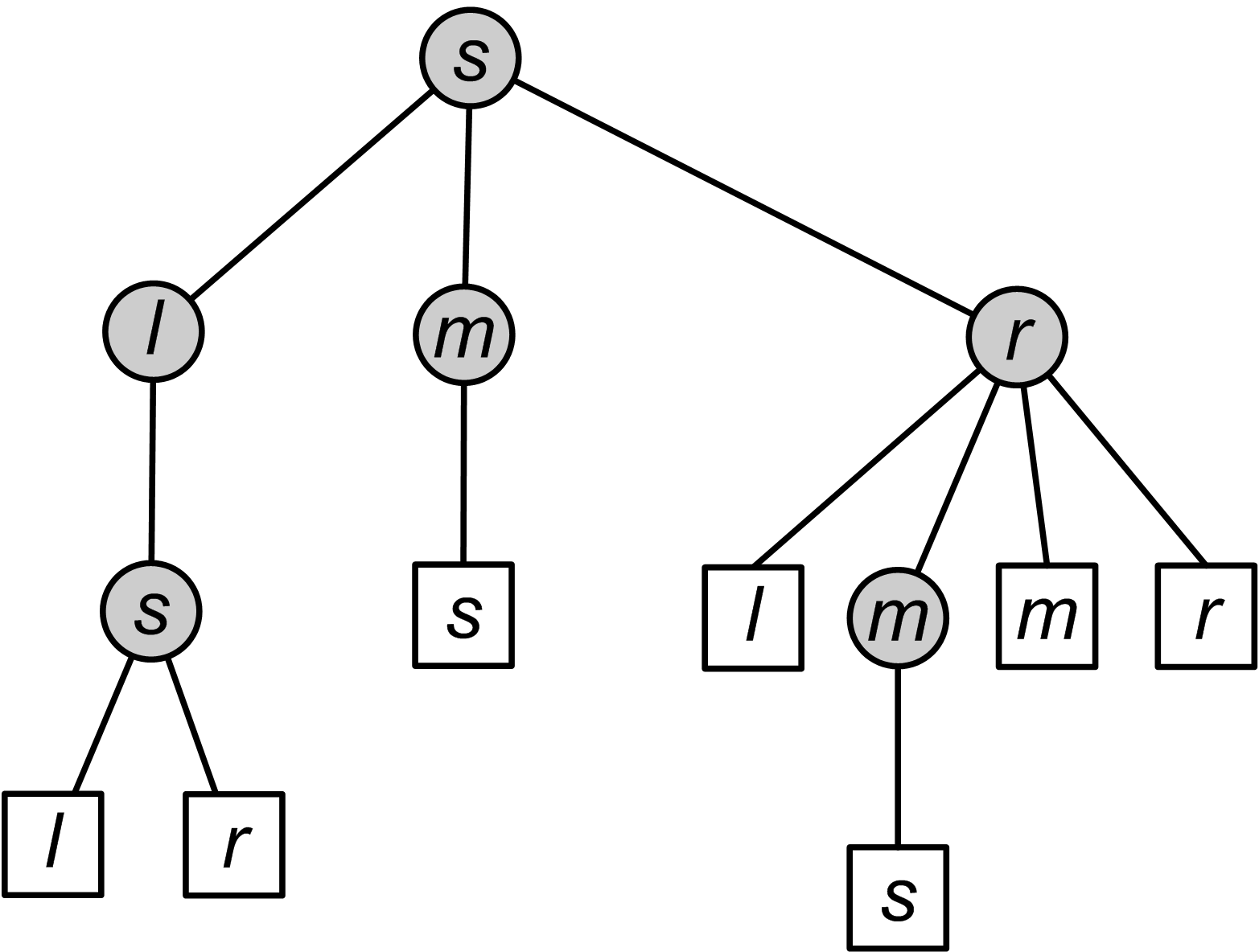}
\end{center}
\caption{An example of a tree $T\in T_n^\Delta$ (for $\Delta \geq 4$). Its codeword is ``$s\ell s\ell rmsr\ell msmr$''.}
\label{chemcoded}
\end{figure}

\begin{figure}
\begin{center}
\includegraphics[scale=0.75]{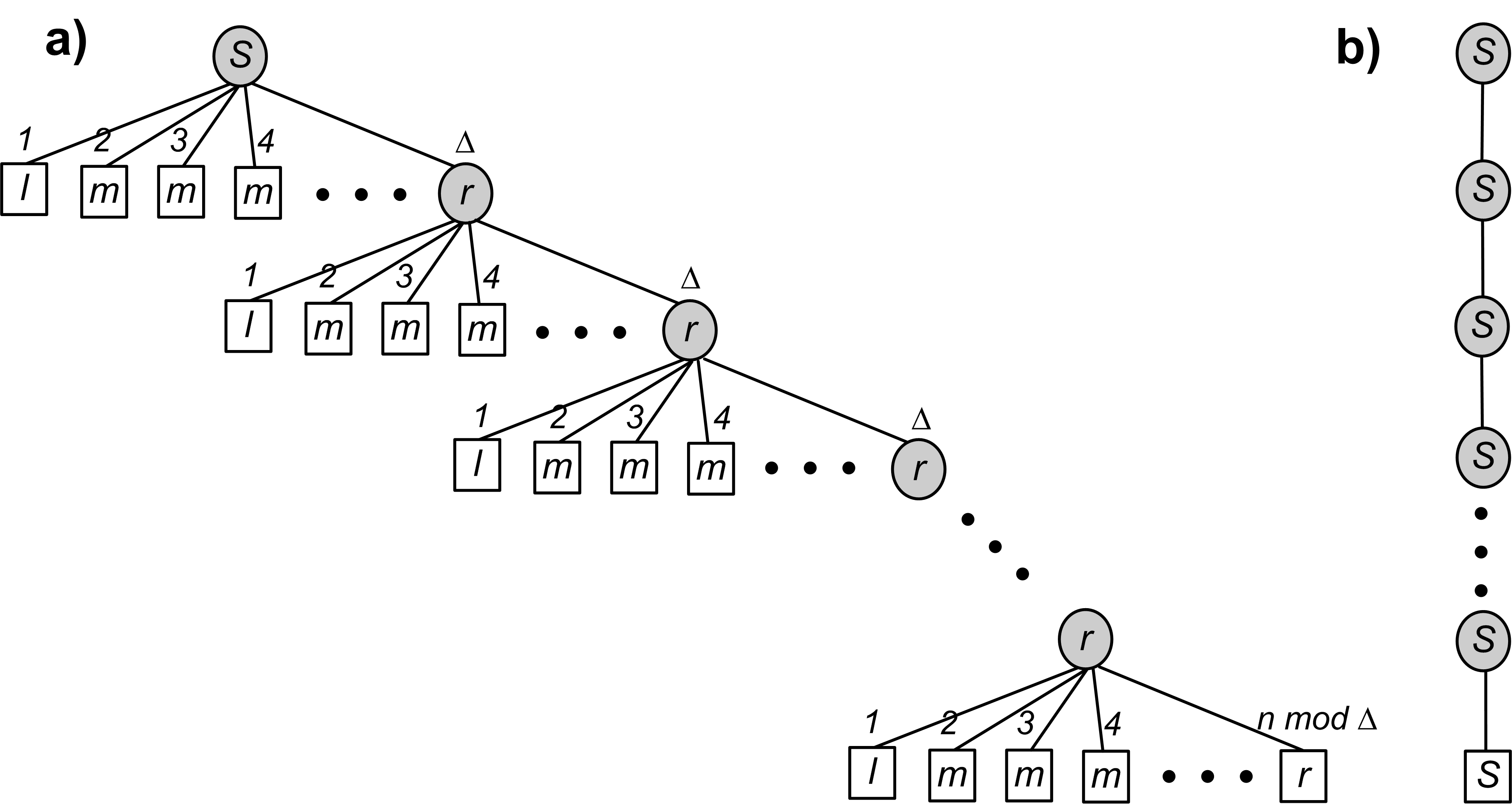}
\end{center}
\caption{ a) The first $T_n^\Delta$ tree in A-order. b) The last $T_n^\Delta$ tree in A-order.}
\label{chemfirstlast}
\end{figure}

\begin{definition}
Suppose that $\{s, \ell , m, r\}^*$ is the set of all sequences with alphabet of $s,  m, \ell ,  r$ and let $A$ be a proper subset of $\{s,\ell , m, r\}^*$, then we call the set $A$ a {\em CodeSet$^\Delta$} iff $A$ satisfies the following properties:

\begin{enumerate}
    \item  $\epsilon \in  A$ \  $(\epsilon$ is a string of length $0)$,
    \item   $\forall x \in A: sx \in A$,
    \item   $\forall \ x_1, x_2, \ldots, x_i \in  A$, and $2\leq i \leq \Delta$: $\ell x_1 m x_2 m x_3 \ldots m x_{i-1} r x_i \in A$.
\end{enumerate}
\end{definition}
Now we show that a {\bf valid codeword} is obtained by the concatenation of the character $s$ and each element of CodeSet$^\Delta$.
\begin{theorem}
Let $A$ be the ``CodeSet$^\Delta$" and $\delta \in A$.  Let $C$ be a codeword obtained by the concatenation of character $s$ and $\delta$  (denoted by $s\delta$). There is a one-to-one correspondence between $C$ and a $T^\Delta$ tree.
\end{theorem}
\begin{proof}{
It can be proved by induction on the length of $C$. Initially for a codeword of length equal to 1, the proof is trivial. Assume that any codeword obtained in the above manner with length less than $n$ encodes a unique $T^\Delta$ tree. For a given codeword with length $n$, because of that concatenation of $s$ and $\delta$, we have:
\begin{enumerate}
    \item $C = sx$, such that $x \in A$, or
    \item  $C = s\ell x_1 mx_2 \ldots mx_{j-1} rx_j$, such that $x_i \in A$, $\forall 1\leq i \leq j \leq \Delta$.
\end{enumerate}

For the first case by induction hypothesis, $x$ is a valid codeword of a $T^\Delta$ tree $T$; therefore, $sx$ is another codeword corresponding to another $T^\Delta$ tree by adding a new root to the top of $T$. This tree is unique and shown in Figure~\ref{chemivalid}-a.
For the second case, by induction hypothesis and that concatenation of $s$ and $\delta$, each $sx_i$ for $1\leq i\leq j$ is a valid codeword for a $T^\Delta$ tree; therefore, with replacement of `$s$ with $\ell$  in $sx_1$' and `$s$ with $m$ in $sx_i$ for $2\leq i\leq j-1$' and finally `$s$ with $r$ in $sx_j$', we can produce $\ell x_1, mx_2, \ldots mx_{j-1}, rx_j$ codewords. Now they all are subtrees of a $T^\Delta$ tree whose codeword is $C = s\ell x_1 mx_2 \ldots mx_{j-1} rx_j$ (add a new root and connect it to each one of them). This tree is unique and shown in Figure~\ref{chemivalid}-b.  
%%In a similar manner, by another induction we can prove that each $T^\Delta$ tree has a unique code in the form $C = s\ell x_1 mx_2 \ldots mx_{j-1} rx_j$. Therefore one-to-one correspondence between codewords and $T^\Delta$ trees holds. 
%Reversely, we assume each $T^\Delta$ tree $T$ with size $k<n$ has a unique code in the form $C = s\ell x_1 mx_2 \ldots mx_{j-1} rx_j$. Now we consider the tree $T$ with size $n$ as shown in Figure~\ref{chemivalid}-b. Employing the induction assumption, the codeword corresponding to the first subtree of this tree is equal to $\ell x_1$ where $x_1\in A$ and the codeword corresponding to the second subtree is equal to $mx_2$ where $x_2\in A$, $\ldots$, and the codeword corresponding to the $j^{th}$ subtree is equal to $rx_j$ where $x_j\in A$. With respect to the definition of $A$, the concatenation of these codewords ($\delta =\ell x_1 mx_2 \ldots mx_{j-1} rx_j$) belongs to $A$. Clearly $C=s\delta$  presents a codeword of a tree composed of the above subtrees.	
}\end{proof}

\begin{figure}[!t]
\begin{center}
\includegraphics[scale=0.73]{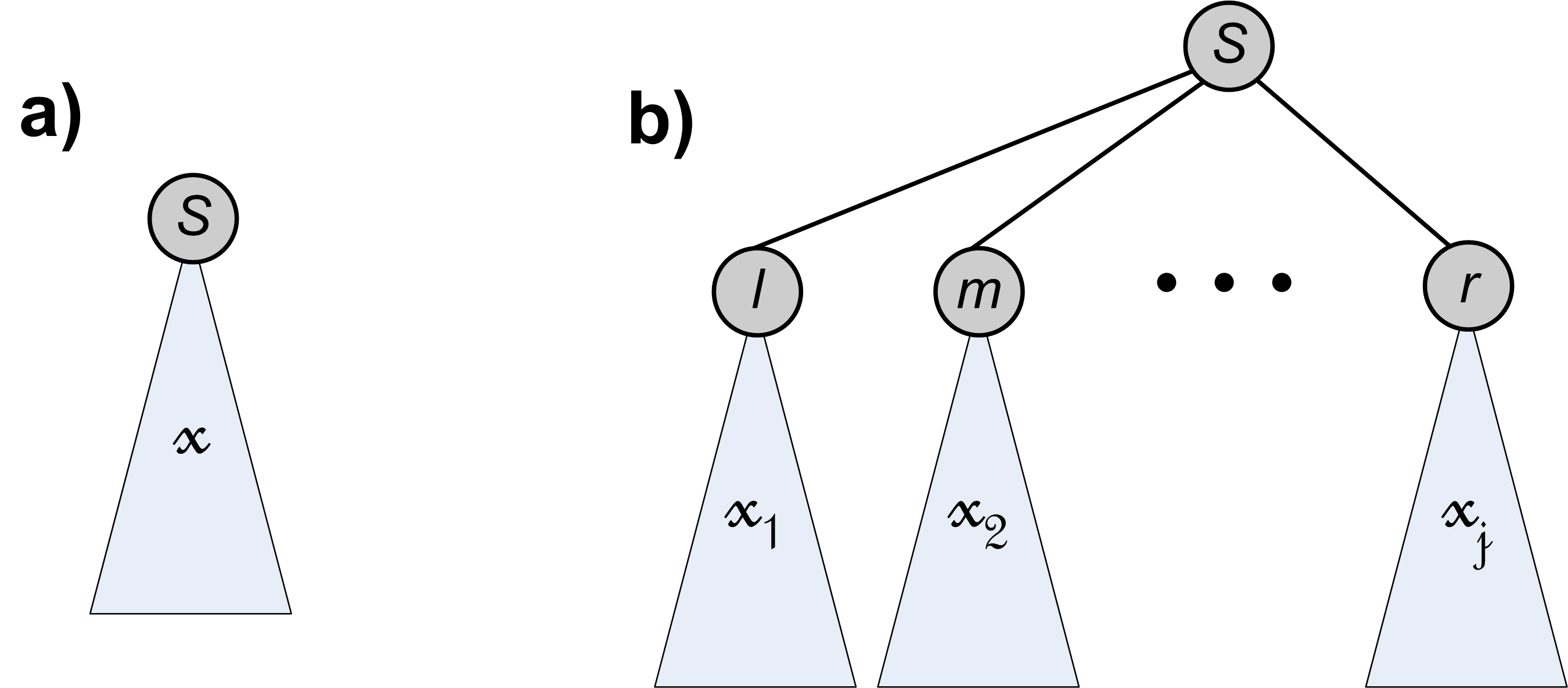}
\end{center}
\caption{$T^\Delta$ trees encoded by $C=sx$ and $C = s\ell x_1 mx_2 \ldots mx_{j-1} rx_j$.}
\label{chemivalid}
\end{figure}

For a $T_n^\Delta$ tree, this encoding needs only 4 alphabet letters and has length $n$. This encoding is simple and powerful, so it can be used for many other applications besides the generation algorithm. In the next section, we use it to generate $T^\Delta$ trees in A-order. %shaiad beshe hazf kard ???

%-------------------------------------------------------------------------------------------------------
\section{The generation algorithm}
%------------------------------------------------------------------------------------------------
\label{nextchap}

\noindent In this section, we present an algorithm that generates the successor
 sequence of a given codeword of a $T_n^\Delta$ tree in A-order.
For generating the successor of a given codeword $C$ corresponding to a $T_n^\Delta$ tree $T$, the codeword $C$ is scanned from right to left. Scanning the codeword $C$ from right to left, corresponds to a reverse pre-order traversal of $T$. 
First we describe how this algorithm works directly on  $T$, then we present the pseudocode of the algorithm. % for generating the successor of $C$.
For generating the successor of a given $T_n^\Delta$ tree $T$ we traverse the tree in reverse pre-order as follows.
\begin{enumerate}
    \item  Let $v$ be the last node of  $T$ in pre-order traversal.
    \item   If $v$ doesn't have any brothers, then 
	\begin{itemize}
     		  \item  repeat $\{v =$ parent of $v.\}$
     		  
%			\begin{itemize}
 %		    		  \item  $\{ v =$ parent of $v.\}$ % \item  $\{ v = \text{ father of } v.\}$ 
%			\end{itemize}
 		        until $v$ has at least one brother or $v$ be the root of tree $T$.
  		  \item   If $v =$ root, then the tree is the last tree in A-order and there is no successor.
	\end{itemize}
    \item   If $v$ has at least one brother (obviously it has to be a left brother), delete one node from the subtree of $v$ and insert this node into its left brother's subtree, then rebuild both subtrees (each one as a first tree with corresponding nodes in A-order).
\end{enumerate}

The pseudo code of this algorithm for codewords corresponding to $T_n^\Delta$ trees is presented in Figure~\ref{genchmical}. In this algorithm, the codeword is stored in a global array $C$ so it is used both as the input and the output codewords (this will be useful to generate all the codewords in constant average time); therefore, when the algorithms stops, $C$ is updated to the successor codeword.
%, $C$ is a global array of characters holding the codeword (the algorithm generates the successor sequence of this codeword), 
In this algorithm, $n$ shows the size of the codeword (the number of nodes of the tree corresponded to $C$), $STsize$ is a variable contains the size of the subtree rooted by node corresponded to $C[i]$ and $SNum$ holds the number of consecutive visited $s$ characters. 
This algorithm also calls two functions $updateChildren( i, ChNum)$ presented in Figure~\ref{udchild}, and $updateBrothers( i, ChNum)$ presented in Figure~\ref{upBro}.

The procedure $updateChildren(i,ChNum)$ regenerates the codeword corresponding to the children of an updated node and the procedure $updateBrothers(i,ChNum)$  regenerates the codeword corresponding to the brothers of a node with regard to the maximum degree $\Delta$ for each node. In these algorithms, $C$ is the global array of characters storing the codeword, $i$ is the position of the current node in the array $C$.
In $updateChildren(i,ChNum)$, $ChNum$ is the number of children of $C[i]$ to regenerate the corresponding codeword and in $updateBrothers(i,ChNum)$,  $ChNum$ is the number of its brothers, instead. 
Each node has at most $\Delta$ children, so $NChild$ is a global array which  $NChild[i]$ stores the number of children of the parent of the current node (current node is the node corresponding to $C[i]$) so it wont exceed $\Delta$, \ie $NChild[i]$ holds the number of left brothers of the node corresponding to $C[i]$ including itself.

\begin{figure}[!t]%[!b]%[!t]
\begin{algorithmic}[0] %baraie shomare zadan
\STATE \hspace*{0.3cm} {\bf Function} {\em AOrder-Next}($n$ :  {\bf integer});
\STATE \hspace*{0.3cm} {\bf var}  $i, Current, STsize, SNum$: {\bf integer}; 								$finished, RDeleted$: {\bf boolean};
\STATE \hspace*{0.3cm} {\bf begin} 
\STATE \hspace*{1cm}   	$Current := n$;    $STSize := 0$;    $RDeleted :=$ false;    						$finished :=$ false;
\STATE \hspace*{1cm}  		{\bf while} ( ($C[Current]=$ $'s'$) $\&$ ($Current \geq 1$) ) do 
\STATE \hspace*{1.7cm}  		  $STSize++$; $Current--$;  
\STATE \hspace*{1cm}    	{\bf if} ($Current=0$) {\bf then} {\bf return} (‘no 								successor’);
\STATE \hspace*{1cm}    	{\bf while} (not $finished$) {\bf do begin} 
%\STATE \hspace*{1cm} 		{\bf begin}
\STATE \hspace*{1.7cm}		$STSize ++$;
\STATE \hspace*{1.7cm}		{\bf switch}  $C[Current]$ {\bf of}
%\STATE \hspace*{1.7cm}		$\{$
\STATE \hspace*{2.4cm}			{\bf case}$'r'$: 
\STATE \hspace*{3.1cm}					$i:=Current-1$;  $SNum:=0$;
\STATE \hspace*{3.1cm}					{\bf while}  ($C[i]=$ $'s'$) {\bf do} \STATE \hspace*{3.8cm}						{ $SNum := SNum + 1;    i--;$}
\STATE \hspace*{3.1cm}					{\bf if}  ($C[i]=$ $'r'$)  {\bf then begin}
\STATE \hspace*{3.8cm}						{\em updateBrothers} ( $Current$ , $STSize$); \STATE \hspace*{3.8cm}                                   $Current:=i$;   $STSize:=SNum$; 
\STATE \hspace*{3.1cm}	                          {\bf end};
\STATE \hspace*{3.1cm}					{\bf if} ( ($C[i]=$  $'m'$) {\bf or} ($C[i]=$ $'\ell '$) ) {\bf then} {\bf begin}	 				
%\STATE \hspace*{3.1cm}					{\bf begin}
\STATE \hspace*{3.8cm}						{\bf if}  ($STSize=1$)  {\bf then}     RDeleted := true;
\STATE \hspace*{3.8cm}						{\bf if}  ($STSize>1$) {\bf then} {\bf begin}
%\STATE \hspace*{3.8cm}						{\bf begin}
\STATE \hspace*{4.5cm}		 					{$STSize--$;    $updateBrothers( Current+1, STSize);$ }
\STATE \hspace*{4.5cm}							$Current := i$;    $STSize := SNum+1;$
\STATE \hspace*{3.8cm}						{\bf end};
\STATE \hspace*{3.1cm}					{\bf end};
%\STATE \hspace*{0.3cm}		}
\STATE \hspace*{2.4cm}			{\bf case} $'m'$: %{
\STATE \hspace*{3.1cm}				{\bf if}  ($RDeleted=true$)  {\bf then} {$C[Current] :=$ $'r';$}
\STATE \hspace*{3.1cm}			 	{\em updateChildren}( $Current+1, STSize -1$);    $finished$:= true;
%\STATE \hspace*{0.3cm}		}
\STATE \hspace*{2.4cm}			{\bf case} $'\ell '$:  %{ 
\STATE \hspace*{3.1cm}				{\bf if}   ($RDeleted=$ true)  {\bf then}  {$C[Current] :=$ $'s';$}
\STATE \hspace*{3.1cm}				 {\em updateChildren}( $Current+1, STSize -1$);   $finished$ := true;
%\STATE \hspace*{0.3cm}		}
%\STATE \hspace*{0.3cm}	}   //switch
\STATE \hspace*{1cm} 		{\bf end};
\STATE \hspace*{0.3cm}{\bf end};
\end{algorithmic}
 \begin{caption}{\label{genchmical}
  Algorithm for generating the successor codeword for $T_n^\Delta$ trees in A-order.}
\end{caption}
\end{figure}

\begin{figure}%[!b]%[!t]
\begin{algorithmic}[0] %baraie shomare zadan
\STATE \hspace*{0.3cm}{\bf procedure} {\em updateChildren}( $i, ChNum$: {\bf integer});
\STATE \hspace*{0.3cm}{\bf begin}
\STATE \hspace*{1cm}		{\bf while} $(ChNum>0)$ {\bf do begin}
\STATE \hspace*{1.7cm}		{\bf if} $ChNum=1$ {\bf then begin}
\STATE \hspace*{2.4cm}			$C[i] :=$ $'s'$; $NChild[i]:=1$; $i++$; $ChNum--;$
\STATE \hspace*{1.7cm}		{\bf end};
\STATE \hspace*{1.7cm}		{\bf if} $ChNum>1$ {\bf then begin}
\STATE \hspace*{2.4cm}			$C[i] :=$ $'\ell'$; $NChild[i]:=1$; $i++$; $ChNum--$; 
\STATE \hspace*{2.4cm}			{\bf while} ( ($NChild[i]<(\Delta-1)$) $\&$ ($ChNum>1$) ) {\bf do begin} 
\STATE \hspace*{3.1cm}				 $C[i]:=$ $'m'$;  $NChild[i]:=NChild[i-1]+1$;   $i++$;  $ChNum--;$
\STATE \hspace*{2.4cm}			{\bf end};
\STATE \hspace*{2.4cm}			$C[i] :=$ $'r'$; $NChild[i]:=NChild[i-1]+1$; $i++$;    $ChNum--$;
\STATE \hspace*{1.7cm}		{\bf end}
\STATE \hspace*{1cm}		{\bf end};
\STATE \hspace*{0.3cm}{\bf end};
\end{algorithmic}
 \begin{caption}{\label{udchild}
  Algorithm for updating the children. }
\end{caption}
\end{figure}

\begin{figure}[!t]%[!b]%[!t]
\begin{algorithmic}[0] %baraie shomare zadan
\STATE \hspace*{0.3cm}{\bf Procedure} {\em updateBrothers}( $i, ChNum$: {\bf integer});
\STATE \hspace*{0.3cm}{\bf begin}
\STATE \hspace*{1cm}		{\bf if} $ChNum=1$ {\bf then begin}
\STATE \hspace*{1.7cm}		$C[i]:=$ $'r'$;   $NChild[i]:=NChild[i-1]$;   $ChNum-- ;$
\STATE \hspace*{1cm}		{\bf end};
\STATE \hspace*{1cm}    	{\bf if} $ChNum>1$ {\bf then begin}
\STATE \hspace*{1.7cm}		$C[i]:=$ $'m'$;   $ChNum--$;   $i++;$
\STATE \hspace*{1.7cm}		{\bf while} ( ($NChild[i]<(\Delta-1)$ ) $\&$ ($ChNum>1$) ) {\bf do begin}
\STATE \hspace*{2.4cm}  			$C[i]:=$ $'m'$;   $NChild[i]:=NChild[i-1]+1$; $i++$;    $ChNum--$;
\STATE \hspace*{1.7cm}		{\bf end};
\STATE \hspace*{1.7cm}		$C[i] :=$ $'r'$;   $NChild[i]:=NChild[i-1]+1$;
\STATE \hspace*{1.7cm}		$ i++$;    $ChNum--$;	  $updateChildren(i, ChNum)$;
\STATE \hspace*{1cm}		{\bf end};
\STATE \hspace*{0.3cm}{\bf end};
\end{algorithmic}
 \begin{caption}{\label{upBro}
  Algorithm for updating the neighbors. }
\end{caption}
\end{figure}
 
\begin{theorem}
The algorithm Next presented in Figure~\ref{genchmical} has a worst case time complexity of $O(n)$ and an average time complexity of $O(1)$.
\label{nexttime}
\end{theorem}
\begin{proof}{
In the algorithm given in Figure~\ref{genchmical}, we scan the sequence from right to left (corresponding to the reverse pre-order traversal of the corresponding tree as discussed before). Inside or after every loop of the algorithm, the variable $Current$ (which keeps track of the node we currently process), decreases equivalently to the number of iterations, and it can not be decreased more than $n$ times; therefore,  the worst case time complexity of the algorithm is $O(n)$.
%  Worst case time complexity of this algorithm is $O(n)$ because the sequence is scanned just once.
   For computing the average time, it should be noted that during the scanning process, every time we visit the characters $s$, $m$ or $\ell$, the algorithm terminates, so we define $S^{n, \Delta}_i$ as the number of codewords of $T_n^\Delta$ trees whose last character $s$, $m$ or $\ell$ has distance $i$ from the end, and $S^{n, \Delta}$ as the total number of $T_n^\Delta$ trees.
 Obviously we have:
\begin{equation}
S^{n,\Delta} = \sum_{i=1}^{n} S^{n,\Delta}_i.
\end{equation}
We define $H_n$ as the average  time of generating all codewords of $T_n^\Delta$ trees, 
\begin{quote}
\begin{quote}
$H_n \leq (k/S^{n,\Delta}) \sum_{i=1}^{n} iS^{n,\Delta}_{i}$,\\
${\ \ \ \ } \leq (k/S^{n,\Delta}) \sum_{j=1}^{n} \sum_{i=j}^{n}S^{n,\Delta}_{i}$.
\end{quote}
\end{quote}
Where $k$ is a constant value. On the other hand, consider that for $S^{n+1,\Delta}_j$  we have two cases, in the first case, the last character $s$, $m$ or $\ell$ is a leaf and in the second one, it is not.
Therefore, $S^{n+1,\Delta}_j$ is greater than or equal to just the first case, and in that case by removing the node corresponding to the `last character $s$, $m$ or $\ell$ of the codeword', the remaining tree will have a corresponding codeword belongs to exactly one of $S^{n, \Delta}_k$ cases, for $j\leq k \leq n$. By substituting $k$ and $i$ we have:
\begin{quote}
\begin{quote}
 $S^{n+1, \Delta}_j \geq \sum_{i=j}^{n} S^{n, \Delta}_i$.
\end{quote}
\end{quote}
Therefore, for $H_n$ we have:
\begin{quote}
\begin{quote}
$H_n \leq (k/S^{n,\Delta}) \sum_{j=1}^{n} S^{n+1, \Delta}_{j}$,
\end{quote}
\end{quote}
then by using Equation (1),
\begin{quote}
\begin{quote}
$H_n \leq k S^{n+1,\Delta}/S^{n,\Delta} $.
\end{quote}
\end{quote}
Finally from~\cite{za80} we know that the total number of ordered trees is growing same as Catalan number, while $T_n^\Delta$ is a subset of ordered trees can not grow faster than that, this guarantees that for large enough values of $n$, $S^{n+1,\Delta}/S^{n,\Delta}=O(1)$. Therefore,
$H_n \leq kO(1)=O(1)$.
 }\end{proof}
It should be mentioned that this constant average time complexity is  without considering the input or the output time. 

%%% ----------------------------------------------------------------------
\section{Ranking and Unranking algorithms}
%%% ----------------------------------------------------------------------
\label{Ranksection}
\noindent By designing a generation algorithm in a specific order, the ranking of algorithm is desired.  In this section, ranking and unranking algorithms for these trees in A-order will be given. %The following theorems and definitions help us in designing the rank algorithm in A-order.
  Ranking and unranking algorithms usually use a precomputed table of the number of a subclass of given trees with some specified properties to achieve efficient time complexities; these precomputations will be done only once and stored in a table for further use. %~\cite{para85,pa87,li86}. 
Let  $S^{n, \Delta}$ be the number of $T_n^\Delta$ trees, 
 $S^{n,\Delta}_{m,d}$  be the number of $T_n^\Delta$ trees whose first subtree has {\bf exactly} $m$ nodes and its root has maximum degree of $d$, and  $D^{n,\Delta}_{m,d}$  be the number of $T_n^\Delta$ trees whose first subtree has {\bf at most} $m$ nodes and its root has maximum degree of $d$. %to store the results of this precomputation. %We use these parameters for ranking and unranking formulas.
 \begin{theorem}
 \label{DfromS}\ 
 \begin{itemize}
    \item  $D^{n,\Delta}_{m,d} =\sum_{i=1}^m S^{n,\Delta}_{i,d}$,
    \item $S^{n,\Delta}=\sum_{i=1}^{n-1} S^{n,\Delta}_{i,\Delta}$.
\end{itemize}
\end{theorem}
\begin{proof}{
The proof is trivial. 
}\end{proof}
%\begin{theorem}
%$S^{n,\Delta}$ is equal to:
%
%$$S^{n,\Delta}=\sum_{i=1}^{n-1} S^{n,\Delta}_{i,\Delta}$$
%\end{theorem}
%\begin{proof}
%The proof is trivial. 
%\end{proof}

\begin{theorem}
\label{snmd}
%$S^{n,\Delta}_{m,d}$ is equal to:

$$S^{n,\Delta}_{m,d}=S^{m+1,\Delta}_{m,1}\times \sum_{i=1}^{n-m-1}(S^{n-m,\Delta}_{i,d-1}).$$
\end{theorem}
\begin{proof}{
Let $T$ be a $T_n^\Delta$ tree whose first subtree has exactly $m$ nodes and its root has maximum degree of $d$; by the definition and as shown in the Figure~\ref{chemiccount}, the number of the possible cases for the first subtree is $S^{m+1,\Delta}_{m,1}$  and the number of cases for the other parts of the tree is
$\sum_{i=1}^{n-m-1}(S^{n-m,\Delta}_{i,d-1})$ (by removing the first subtree from $T$).
Therefore, we have:
$$S^{n,\Delta}_{m,d}=S^{m+1,\Delta}_{m,1}\times \sum_{i=1}^{n-m-1}(S^{n-m,\Delta}_{i,d-1}).$$
}\end{proof}

\begin{figure}[!t]
\begin{center}
\includegraphics[scale=0.95]{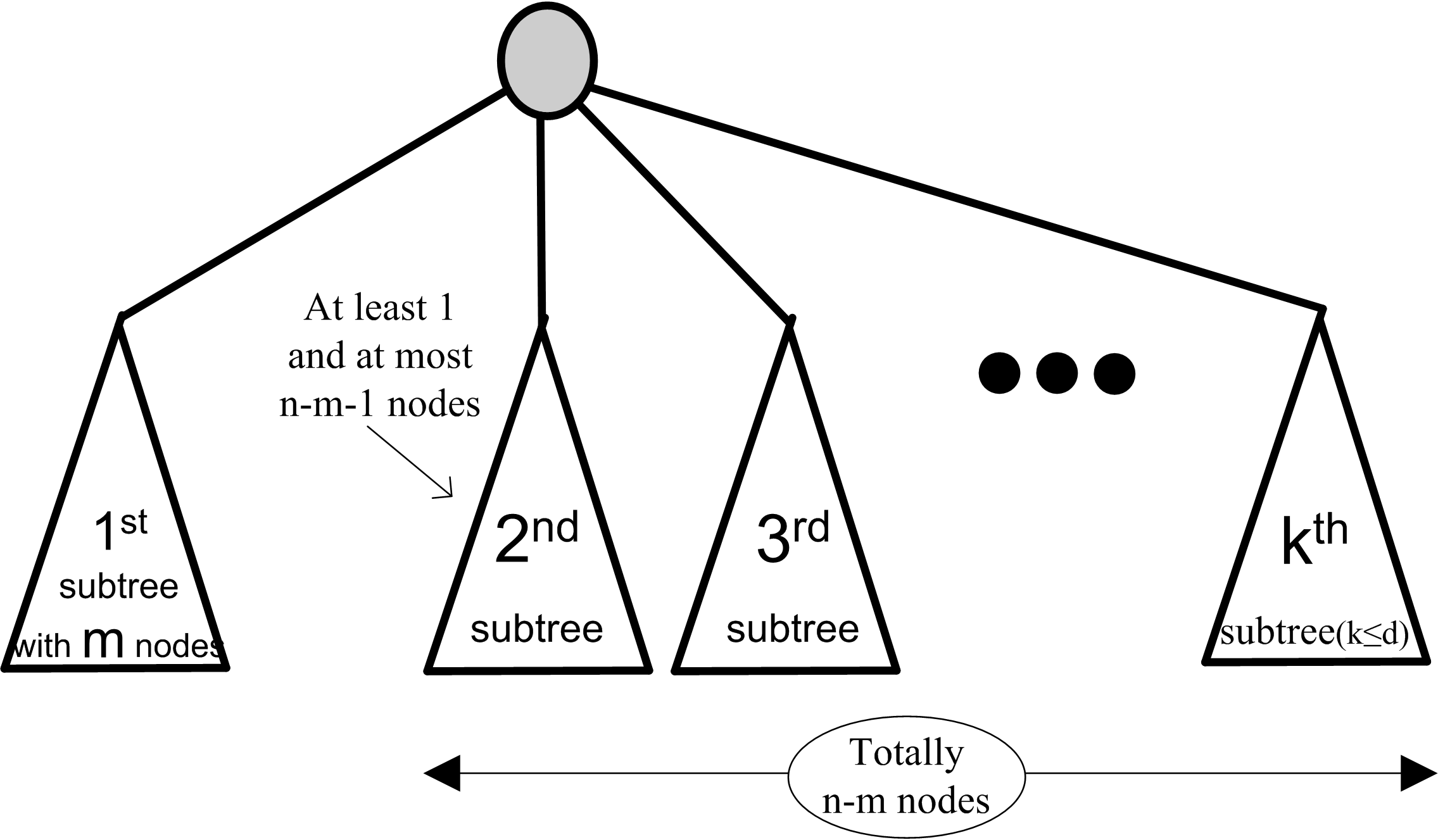}
\end{center}
\caption{$T_n^\Delta$ tree whose first subtree has exactly $m$ nodes and its root has maximum degree of $d$.}
\label{chemiccount}
\end{figure}

 Now, let $T$ be a $T_n^\Delta$ tree whose subtrees are defined by $T_1, T_2, \ldots , T_k$ and for $1\leq i \leq k\leq \Delta$ : $|T_i|=n_i$ and $\sum_{i=1}^{k}n_i=n-1$.  
One way to compute the rank of tree $T$ is to enumerate the number of trees generated before $T$. 
 % %For computing the rank of $T$, we have to enumerate the number of trees generated before $T$. 
  Let {\em Rank(T,n)} be the rank of $T$. The number of $T^\Delta$ trees whose first subtree is smaller than $T_1$ is equal to:
$$\sum_{i=1}^{n_1-1}S^{n,\Delta}_{i,\Delta} + (Rank(T_1,n_1)-1)\times \sum_{i=1}^{n-n_1}S^{n-n_1,\Delta}_{i,\Delta-1},$$

\noindent and the number of $T^\Delta$ trees whose first subtree is equal to $T_1$ but the second subtree is smaller than $T_2$ is equal to:
$$\sum_{i=1}^{n_2-1}S^{n-n_1,\Delta}_{i,\Delta-1} + (Rank(T_2,n_2)-1)\times \sum_{i=1}^{n-n_1-n_2}S^{n-n_1-n_2,\Delta}_{i,\Delta-2}.$$

\noindent Similarly, the number of $T^\Delta$ trees whose first ($j-1$) subtrees are equal to $T_1, T_2, \ldots , T_{j-1}$ and the $j^{th}$ subtree is smaller than $T_j$ is equal to:

$$\sum_{i=1}^{n_j-1}S^{(n-\sum_{\ell=1}^{j-1}n_\ell),\Delta}_{i,\Delta-j+1} + (Rank(T_j,n_j)-1)\times \sum_{i=1}^{n-\sum_{\ell=1}^{j}n_\ell}S^{n-\sum_{\ell=1}^{j}n_\ell,\Delta}_{i,\Delta-j}.$$
 Therefore, regarding enumerations explained above, for given tree $T\in T_n^\Delta$  whose subtrees are defined by $T_1, T_2, \ldots , T_k$ , we can write:
\begin{eqnarray*}
Rank(T,1)&=&1,\\
 Rank(T,n)&=&1+\sum_{j=1}^k ( \sum_{i=1}^{n_j-1} S^{(n-\sum_{\ell=1}^{j-1}n_\ell),\Delta}_{i,\Delta-j+1}+ 
( Rank(T_j, n_j)-1) \sum_{i=1}^{n-\sum_{\ell=1}^{j}n_\ell} S^{(n-\sum_{\ell=1}^{j}n_\ell),\Delta}_{i,\Delta-j}).
\end{eqnarray*}
Hence, from Theorem~\ref{DfromS}, by using $D^{n,\Delta}_{m,d} =\sum_{i=1}^m S^{n,\Delta}_{i,d}$, we have:
\begin{eqnarray*}
Rank(T,1)&=&1,\\
 Rank(T,n)&=&1+\sum_{j=1}^k ( D^{(n-\sum_{\ell=1}^{j-1}n_\ell),\Delta}_{(n_j-1),(\Delta-j+1)}+ 
 ( Rank(T_j, n_j)-1) D^{(n-\sum_{\ell=1}^{j}n_\ell),\Delta}_{(n-\sum_{\ell=1}^{j}n_\ell),(\Delta-j)}).\ \ \ \ \ \  \ \ \ \ 
\end{eqnarray*}

To achieve the most efficient time for ranking and unranking algorithms, we need to precompute $D^{n,\Delta}_{m,d}$ and store it for further use. 
 Assuming $\Delta$ is constant, to store $D^{n,\Delta}_{m,d}$ values, a 3-dimensional table denoted by $D[n,m,d]$ is enough, % (corresponding to $n$, $m$, and $d\le \Delta$), 
 this table will have a size of $O(n\times n \times \Delta)=O(n^2)$ and can be computed using Theorems~\ref{DfromS} and~\ref{snmd} with time complexity of $O(n\times n \times \Delta)=O(n^2)$. %We assume $D^{n,\Delta}_{m,d}$  is stored in table $D[n,m,d]$ and is computed only once as pre-processing step of ranking and unranking algorithms.

To compute the rank of a codeword stored in array $C$, we also need an auxiliary array $N[i]$ which keeps the number of nodes in the subtree whose root is labeled by $C[i]$ and corresponds to $n_i$ in the above formula. This array can be computed  by a pre-order traversal or a level first search (DFS) algorithm  just once before  we call the ranking algorithm.

%For presenting the ranking algorithm we need another auxiliary array $D[n, \Delta, m, d]$ which is equal to the number of $T_n^\Delta$ trees whose  first subtree has at most $m$ nodes and its root has a maximum degree of $d$. Obviously we have:  $D[n,\Delta, m, d] =\sum_{i=1}^m S[n,i,d]$.

The  pseudo code  for ranking algorithm is given in Figure~\ref{chemicrank}. In this algorithm, $Beg$ is the variable that shows the positions of the first character in the array $C$ whose rank is being computed ($Beg$ is initially set to 1), and $Fin$ is the variable that returns the position of the last character of $C$.

\begin{figure}[!t]%[!b]%[!t]
\begin{algorithmic}[0] %baraie shomare zadan
\STATE \hspace*{0.3cm}{\bf Function} {\em Rank}( $Beg:$ {\bf integer}; {\bf var} $Fin$: {\bf integer}) ;
\STATE \hspace*{0.3cm}{\bf Var}   $R$, $Point$, $PointFin$, $j$, $Nodes$, $n$: {\bf integer};
\STATE \hspace*{0.3cm}{\bf begin}
\STATE \hspace*{1cm}		$n := N [Beg] ;$
\STATE \hspace*{1cm}		{\bf if} ($n = 1$) {\bf then begin} 
\STATE \hspace*{1.7cm}		{ $Fin:=Beg$;   {\bf return}(1) } {\bf end};	
\STATE \hspace*{1cm}		{\bf else begin}
\STATE \hspace*{1.7cm}		 $Point:= Beg+1$;   $R:= 0$;   $Nodes:= 0$; $j:=1$;
\STATE \hspace*{1.7cm}		{\bf while} ( $Nodes<n$ ) {\bf do begin} 
\STATE \hspace*{2.4cm}			$R:=R+D[n-Nodes,N [Point]-1,\Delta-j+1] + $
\STATE \hspace*{3.3cm}				$(Rank(Point,PointFin)-1)\times$
\STATE \hspace*{3.3cm}				$ D[(n-Nodes-N[Point]), (n-Nodes-N[Point]), \Delta-j];$
\STATE \hspace*{2.4cm}			$Nodes := Nodes + N [Point]$; j:=j+1;
\STATE \hspace*{2.4cm}			$Point := PointFin + 1$;
\STATE \hspace*{1.7cm}		{\bf end};
\STATE \hspace*{1.7cm}		$Fin:=Point-1$;
\STATE \hspace*{1.7cm}		{\bf return}( $R+1$);
\STATE \hspace*{1cm}   	{\bf end};
\STATE \hspace*{0.3cm}{\bf end}
\end{algorithmic}
 \begin{caption}{\label{chemicrank}
  Ranking algorithm for $T_n^\Delta$ trees.}
\end{caption}
\end{figure}

Now the time complexity of this algorithm is discussed. Obviously computing the array $N[i]$ takes $O(n)$. Hence we discuss the complexity of  ranking algorithm which was given in Figure~\ref{chemicrank}.
%Since this algorithm is performed just once before calling the ranking algorithm and it has no more time effects on ranking algorithm, we should only calculate the time complexity of ranking algorithm.

\begin{theorem}
The ranking algorithm has the time complexity of $O(n)$.
\label{ranktime}
\end{theorem}
\begin{proof}{
Let $T$ be a $T_n^\Delta$ tree whose
subtrees are defined
 by $T_1$, $T_2$, $\ldots$, $T_k$ and for
 $1 \leq i \leq k\leq\Delta:\ |T_i| = n_i$ and $\sum_{i=1}^{k} n_i = n-1$, and let $T(n)$ be the time
 complexity of the ranking algorithm, then we can write:
 $$ T(n) = T(n_1) + T(n_2) + \ldots + T(n_k) + \alpha  k,$$
where $\alpha$ is a constant and $\alpha k$ is the time complexity of the non-recursive parts of the algorithm. 
 By using a simple induction, we prove that if $\beta$ is a value greater than $\alpha$  then $T(n) \leq \beta  n$.
 We have $T(1) \leq \beta$. We assume $T(m) \leq \beta  (m-1)$ for each $m < n$; therefore,
\begin{quote}
\begin{quote}
$T(n) \leq \beta(n_1-1) + \beta(n_2-1) + \ldots + \beta(n_k-1) + \alpha k$,\\
$T(n) \leq \beta(n_1+ \ldots +n_k - k) + \alpha k$,\\
$T(n) \leq \beta n - \beta  k + \alpha  k \leq \beta n$,\\
%$T(n) \leq \beta n$.\\
\end{quote}
\end{quote}
So the induction is complete and $T(n) \leq \beta n = O(n)$.
% $$T(n) \leq \beta n \Longrightarrow T(n) = O(n).$$ 
}\end{proof}

If $a$ and $b$ are integer numbers, let $(a \ div\  b)$ and $(a\  mod\  b)$ denote {\em integer division} and {\em remainder} of the division of $a$ and $b$, respectively ($a=(a \ div\  b)\times b + (a \ mod\ b)$). Before giving the description of the unranking algorithm, we have to define two new operators, namely $(a\ div^+\ b)$ and $(a\ mod^+ \ b)$ as follows. 
  %%We define  and  as follows.
%Before giving the description of the unranking algorithm we need to
% define two new operators.
 \begin{itemize}
 \item  If $b \mid a$,  then $(a \ div^+ \ b) = (a \ div \ b)-1$, and $(a\ mod^+\ b) = b$.
 \item  If $b \nmid a$, then $(a \ div^+ \ b) = (a \ div \ b)$, and $(a\ mod^+\ b) = (a\ mod\ b)$.
\end{itemize}
For unranking algorithm, we also need the values of $S^{n, \Delta}$, these values can be stored in an array of size $n$, denoted by $S[n]$ (we assume $\Delta$ is constant).
The unranking algorithm is the reverse manner of ranking algorithm, this algorithm is given in Figure~\ref{ChemicUnrank}. In this algorithm, the rank $R$ is the main input, $Beg$ is the variable to show the position of the first character in the global array $C$ and initially is set to 1. The generated codeword will be stored in global array $C$. The variable $n$ is the number of nodes and $Root$ stores the character corresponding to the node we consider for the unranking procedure.
% which we want to compute the unrank of subtree rooted by this node and initially is labeled by $s$ (the label of the root).
For the next character, we have two possibilities. If the root is $r$ or $s$, then the next character, if exists, will be $\ell$ or $s$ (based on the number of root's children). If the root is $m$ or $\ell$,  we have again two possible cases: if all the nodes of the current tree are not produced, then the next character is $m$, otherwise, the next character will be $r$. % all the nodes have been produced and then the next character will be $r$.

\begin{figure}[!t]%[!b]%[!t]
\begin{algorithmic}[0] %baraie shomare zadan
\STATE \hspace*{0.3cm}{\bf Function} {\em UnRank} ( $R, Beg, n$: {\bf integer};   $Root$: {\bf char});
\STATE \hspace*{0.3cm}{\bf var}   $Point$, $i$, $t$, $ChildNum$:  {\bf integer};
\STATE \hspace*{0.3cm}{\bf begin}
\STATE \hspace*{1cm}		{\bf if} ( ($n = 0$) {\bf or} ($R = 0$) ) {\bf then}  {\bf return}($Beg-1$)
\STATE \hspace*{1cm}		{\bf else begin}
\STATE \hspace*{1.7cm}		{\bf if} ($n = 1$) {\bf then begin}
\STATE \hspace*{2.4cm}			$C[Beg] := Root$;   {\bf return}($Beg$);  
\STATE \hspace*{1.7cm}		{\bf end};
\STATE \hspace*{1.7cm}		{\bf	else begin}
\STATE \hspace*{2.4cm}			$C [Beg] := Root$;  $Point := Beg + 1$;
\STATE \hspace*{2.4cm}			$Root :=$  $'\ell'$;   $ChildNum := 0$;
\STATE \hspace*{2.4cm}			{\bf while} ($n > 0$) {\bf do begin} 
\STATE \hspace*{3.1cm}				$ChildNum++$;
\STATE \hspace*{3.1cm}				{\bf find the smallest} $i$ {\bf that} $D[n,  i, \Delta-ChildNum+1] \geq R$;
%\STATE \hspace*{3.1cm}				{\bf repeat}
%\STATE \hspace*{3.8cm}					$i ++;$
%\STATE \hspace*{3.1cm}				{\bf until} ($D[n,  i, \Delta-ChildNum+1] \geq R$);
\STATE \hspace*{3.1cm}				$R := R - D[n, i-1, \Delta-ChildNum+1]$;
\STATE \hspace*{3.1cm}				{\bf if} ($n - i$)$=1$ {\bf then} 
\STATE \hspace*{3.8cm}					{\bf if}  ($ChildNum=1$)   {\bf then}  {   $Root :=$ $'s'$;  }
\STATE \hspace*{3.8cm}					{\bf else}    {    $Root :=$ $' r'$;  }
\STATE \hspace*{3.1cm}				$t := S[n]$;
\STATE \hspace*{3.1cm}				$Point$ := {\em UnRank}( $(div^{+}(R, t))+1,  Point, i, Root$ ) $+ 1$;
\STATE \hspace*{3.1cm}				$R := mod^{+}(R, t)$;
\STATE \hspace*{3.1cm}				$n := n - i$;   $Root :=$ $'m'$;
\STATE \hspace*{2.4cm}			{\bf end};
\STATE \hspace*{2.4cm}			{\bf return}($Point-1$);
\STATE \hspace*{1.7cm}		{\bf end};
\STATE \hspace*{1cm}		{\bf end};
\STATE \hspace*{0.3cm}{\bf end}
\end{algorithmic}
 \begin{caption}{\label{ChemicUnrank}
  Unranking algorithm for $T_n^\Delta$ trees. }
\end{caption}
\end{figure}

\begin{theorem}
The time complexity of the unranking algorithm is $O(n \log n)$.
\label{unranktime}
\end{theorem}
\begin{proof}{
Let $T$ be a $T_n^\Delta$ tree whose
 subtrees are defined
 by $T_1$, $T_2$, $\ldots$, $T_k$ and for
 $1 \leq i \leq k\leq \Delta:\ |T_i| = n_i$ and $\sum_{i=1}^{k} n_i = n-1$, and let $T(n)$ be the time
 complexity of the unranking algorithm.
 With regards to the  unranking algorithm, the time complexity of finding $j$ such
 that $D[n, j, \Delta-ChildNum+1] \geq R$ for each $T_i$ of $T$ is $O(\log n_i)$;
 therefore, 
 \begin{quote}
\begin{quote}
$T(n) = O(\log n_1 + \log n_2 + \ldots + \log n_k) + T(n_1) + T(n_2)
+ \ldots + T(n_k).$
\end{quote}
\end{quote}

We want to prove that $T(n) = O(n\log n )$. In order to obtain an
upper bound for $T(n)$ we do as follows. First we prove this
assumption for $k=2$ then we generalize it. For $k=2$ we have $ T(n)
= O(\log(n_1)+\log(n_2)) + T(n_1) + T(n_2) $. Let $n_1 = x$ then we
can write the above formula as
\begin{quote}
\begin{quote}
$ T(n) = T(x) +
T(n-x) + O(\log(x)+\log(n-x)) =T(x) + T(n-x) + C'\log(n) .$
\end{quote}
\end{quote}
For proving that $T(n) = O(n\log(n))$, we use an induction on $n$. We
assume $T(m) \leq Cm\log(m)$ for all $m \leq n$, thus in $T(n)$ we
can substitute
\begin{quote}
\begin{quote}
$T(n) \leq C \times x\log(x) + C \times (n-x)\log(n-x) + C'\log(n). $
\end{quote}
\end{quote}
Let $f(x) = C \times x\log(x) + C \times (n-x)\log(n-x)$, now the
maximum value of $f(x)$ with respect to $x$ and by considering $n$ as
a constant value can be obtained by evaluating the derivation of
$f(x)$ which is $ f'(x) = C \times \log(x) - C \times \log(n-x) $.
Thus if $ f'(x) = 0$ we get $x = (n-1)/2 $ and by computing $f(1)$,
$f(n-2)$ and $f((n-1)/2)$ we have:
\begin{quote}
\begin{quote}
$f(1) = f(n-2) = C \times (n-2)\log(n-2), $\\
$f((n-1)/2)= 2C \times ((n-1)/2) \times \log((n-1)/2) \ < C \times
(n-2)\log(n-2),$
\end{quote}
\end{quote}
so the maximum value of $f(x)$ is equal to $ C \times (n-2)\log(n-2)$; therefore,
\begin{quote}
\begin{quote}
$  T(n) \leq C \times (n-2)\log(n-2) + C' \times \log(n). $
\end{quote}
\end{quote}
 %\noindent
It is enough to assume $C = C'$, then $T(n) \leq C \times (n-2)\log(n) + C \times \log(n) \ \leq C \times n\log(n).$ 

%\noindent
 Now, for generalizing the above proof and proving $T(n) = O(n \log n)$, we should find the maximum of the
function $f(n_1, n_2, \ldots, n_k) = \prod_{i=1}^{k} n_i$. By the
Lagrange method we prove that the maximum value of $f(n_1, n_2, \ldots,
n_k)$ is equal to
$(\frac{n}{k})^k$. %$ max( f(n_1,n_2, \ldots, n_j) ) = (\frac{n}{j})^j$.
Then $\frac{\delta f}{\delta k} = (\frac{n}{k})^k
(\ln(\frac{n}{k})-1) = 0$, and
\begin{quote}
\begin{quote}
$ \ln(\frac{n}{k}) - 1 = 0$,\\
$ \frac{n}{k} = e \Rightarrow k = \frac{n}{e}$,
\end{quote}
\end{quote}
\noindent
 so the maximum value of $f(n_1, n_2, \ldots, n_k)$ is equal
to $e^{\frac{n}{e}}$. 
We know that:
\begin{quote}
\begin{quote}
$T(n) =O(\log n_1 + \log n_2+ \ldots +\log n_k) + T(n_1) + T(n_2) +\ldots + T(n_k),$

%\noindent
%so

%\begin{quote}
%\begin{quote}
$T(n) = O(\log(\prod_{i=1}^{k}n_i) + \sum_{i=1}^{k}T(n_i)$,\\
$T(n) < O(\log (n^{\frac{n}{e}} ) ) + \sum_{i=1}^{k}T(n_i)$,\\
$T(n) < O(\frac{n}{e} \log e ) = O(n) + \sum_{i=1}^{k}T(n_i)$.
\end{quote}
\end{quote}
\noindent Finally, by using induction, we assume that for any $m<n$ we have $T(m) < \beta m \log m$, therefore:
\begin{quote}
\begin{quote}
$ T(n) = O(n) + \sum_{i=1}^{k}T(n_i) $,\\
$ T(n) < O(n) + \sum_{i=1}^{k} \beta O(n_i\log n_i)$,\\
$ T(n) < O(n) + \beta \log(\prod_{i=1}^{k}(n_i^{n_i}))$,\\
$ T(n) < O(n) + O(\log (n^n))$, \\
$T(n) = O(n\log n)$.
\end{quote}
\end{quote}
\noindent
 Hence, the proof is complete. 
 }\end{proof}

%----------------------------------------------------------------------------------------------
\section{Conclusion}
%----------------------------------------------------------------------------------------------

\noindent In this paper, we studied the problem of generation, ranking and unranking of ordered trees of size $n$ and maximum degree $\Delta$ which are a generalization of chemical trees; we presented an efficient algorithm for the generation of these trees in A-order with an encoding over 4 letters and size $n$. Also two efficient ranking and unranking algorithms were designed for this encoding. The generation algorithm has $O(n)$ time complexity in worst case and $O(1)$ in average case. The ranking and unranking algorithms have $O(n)$ and $O(n \log n)$ time complexity, respectively. The presented ranking and unranking algorithms use a precomputed table of size $O(n^2)$ (assuming $\Delta$ is constant). All presented algorithms at this paper are also implemented. For the future works, generating this class of trees in B-order and minimal change ordering, and finding some explicit relations for counting them, are major unresolved problems.
%For the future works, generating this class of trees in B-order and minimal change
%ordering and finding some explicit relations for enumerating them, are major unresolved related problems.

%\nocite{*}
\bibliographystyle{eptcs}
\bibliography{z-Refrences.bib}

\end{document}